%% file: paper.tex
\documentclass{article}

\usepackage{amsmath, amssymb, verbatim}
\usepackage{amsthm}
\usepackage{graphicx}
\usepackage{fancyhdr}
\usepackage{algorithmic}
\usepackage{algorithm}
\newtheorem{problem}{Problem}
\newtheorem{defn}{Definition}
\newtheorem{theorem}{Theorem}

\newtheorem{fact}{Fact}
\newtheorem{lemma}{Lemma}
\newtheorem{claim}{Claim}
\newtheorem{cor}{Corollary}

\DeclareMathOperator{\prob}{Pr}
\DeclareMathOperator{\Var}{Var}
\DeclareMathOperator{\expect}{\mathbb{E}}

\DeclareMathOperator{\median}{median}

\DeclareMathOperator{\HH}{HH}
\DeclareMathOperator{\polylog}{polylog}
\DeclareMathOperator{\poly}{\mathbf{poly}}
\DeclareMathOperator{\Alg}{\bf{Alg}}

\DeclareMathOperator{\skp}{ sk(\vect{a})}
\DeclareMathOperator{\sk}{sk(\vect{a})}
\DeclareMathOperator{\sfp}{sf}

\newcommand{\skpriv}{\skp^{priv}}
\newcommand{\abs}[1]{\left| #1\right|}
\newcommand{\vect}[1]{\mathbf{#1}}
\newcommand{\norm}[1]{\left | \left| #1 \right|\right|}
\newcommand{\infinorm}[1]{\left|\left|#1\right|\right|_\infty}
\newcommand{\D}{D}
\newcommand{\Sd}{\mathcal{P}}
\newcommand{\eps}{\epsilon}
\newcommand {\Uni}{\mathcal{U}}
\newcommand{\junk}[1]{}
\newcommand{\st}[1]{a_{#1}}
\newcommand{\ai}{\st{i}}
\newcommand{\Dt}{\D^{(t)}}
\newcommand{\vx}{\vect{x}}
\newcommand{\va}{\vect{a}}
\newcommand{\vq}{\vect{q}}
\newcommand{\Dappr}{\tilde{\mathcal{D}} }

\begin{document}

\title{Pan-private Algorithms: When Memory Does Not Help}
\author{Darakhshan Mir\thanks{mir@cs.rutgers.edu. Rutgers University. Work supported by NSF Awardd Number CCF-0728937, and U.S. DHS  Award Number 2009-ST-061-CCI002,}
  \and S.~Muthukrishnan\thanks{muthu@cs.rutgers.edu. Rutgers
    University. Work supported by NSF awards 0354690, 0414852 and 0916782, and DHS CCICADA.} \and Aleksandar
  Nikolov\thanks{anikolov@cs.rutgers.edu. Rutgers University}
  \and Rebecca N.~Wright\thanks{rebecca.wright@rutgers.edu. Rutgers University}}
\maketitle
\input{abs}

\newpage
\section{Introduction}
\input{intro}

\input{sect2}

\input{sect3}

\input{section4}

\input{extensions}

\input{sect5}

\newpage
\bibliographystyle{plain}	
\bibliography{streaming}

\newpage
\end{document}

%% file: abs.tex
\begin{abstract}
Consider updates arriving online in which the $t$th input is $(i_t,d_t)$, where $i_t$'s are thought of as IDs
of users. Informally, a randomized function $f$ is {\em differentially private} with respect to the IDs if the probability distribution induced by $f$ is not much different from that induced by it on an input in 
which occurrences of an ID $j$ are replaced with some other ID $k$. Recently, this notion was  extended to  {\em pan-privacy} where the computation of $f$ retains differential privacy, even if the internal memory of the algorithm is exposed to the adversary (say by a malicious break-in or by fiat by the government). This is a strong notion of privacy, and surprisingly, for basic counting tasks such as distinct counts, heavy hitters and others, Dwork et al~\cite{dwork-pan} present pan-private algorithms with reasonable accuracy. The pan-private algorithms are nontrivial, and rely on sampling.

We reexamine these basic counting tasks and show improved bounds. In particular, 
we estimate the distinct count $\Dt$ to within $(1\pm \eps)\Dt \pm O(\polylog m)$, where $m$ is the number of elements in the universe. This uses suitably noisy statistics on sketches known in the streaming literature. 
We also present the first known lower bounds for pan-privacy with respect to a single intrusion. Our lower bounds show that, even if allowed to work with unbounded memory, pan-private algorithms for distinct counts 
can not be significantly more accurate than our algorithms. Our lower bound uses noisy decoding. 
For heavy hitter counts, we present a pan private streaming algorithm that is accurate to within $O(k)$ in 
worst case; previously known bound for this problem is arbitrarily worse. An interesting aspect of our pan-private algorithms is that, they deliberately use very small (polylogarithmic) space and tend to be streaming algorithms, even though using more space is not forbidden.
\end{abstract}

%% file: intro.tex
Consider updates arriving online in which the $t$th input is $(i_t,d_t)$. 
Define input $S_t$ as the first $t$ updates, i.e.~$(i_1, d_1), \ldots, (i_t, d_t)$; $i_t$'s are IDs
of users from Universe $\Uni$  of size $m$.
%An update $(i_t, d_t)$ is interpreted to mean ``add value $d_t$ to user $i_t$ belonging to a Universe $\Uni$  of size $m$''  
%At any point in time, the state of the system is decsribed by a state vector $\vect{a} =\vect{a}[1] \ldos \vect{a}[m]$, where $\vect{a}[i]$ is the sum of updates seen so far by user $i$. 
An example is to think of this as a ``traffic log'' where $i_t$ is the
ID of a  user, $i_t \in \Uni$ and $d_t$ is the time spent by the user
at a particular website of interest; another example is to think of the input as a ``payment log'' where $i_t$ is the ID of a merchant, $i_t \in \Uni$ and $d_t$ the transaction value, which may be positive for sales and negative for refunds. A user may visit the site many times and a merchant may have many transactions; hence, same $i_t$'s may be seen several times. It is of great interest to maintain various statistics on such logs. For example, 
\begin{itemize}
\item
{\em distinct count}, $\Dt$, is the number of distinct $i_j$'s seen before $t$h update; 
\item
{\em heavy hitters count} $HH(k)$, informally, is the number of $i$'s that have large total $d$, $\sum_{j \leq t| i_t=i} d_j$ (a precise definition is presented later); 
\item
{\em rarity ratio} $r(k)= \frac{|\{i | (\sum_{j\leq t| i_j=i} d_j ) = k \}|}{\Dt}$;
\item
{\em frequency moment} $F_k= \sum_{i \in \Uni}{(\sum_{j  \leq t| i_j=i} d_j)^k}$;
\end{itemize}
and others. Normally, these statistics are trivial to maintain with an
array $\va$ of size $m$ ($a_i = \sum_{j \leq t| i_j = i}{d_j}$),
%$\vect{a}$ of size $|\Uni|$, (essentially the state vector), 
and some basic bookkeeping. These statistics -- in one form or the other -- have a long history, and 
are considered basic in data analysis tasks over the past few decades.

Our focus is on privacy, that is, how to maintain these statistics and still preserve the privacy of IDs involved.  There are two concerns: 
\begin{itemize}
\item
{\em What if the output reveals something about the IDs?} For example, an adversary might estimate $\Dt$ first, and then insert an $(i, 1)$ before determining $D_{t+1}$ which will surely reveal if $i$ was already in the input prior to $t$. Likewise one can devise insertion and query strategies that will reveal information about various IDs from other statistical queries.

\item
{\em What if the adversary gets access to the system and sees the internal memory used by the algorithm?} This might happen not only with intruders but may even be the outcome of a legal request which will force us to reveal all the stored information. In this case, with the trivial solution, $\ai$ will end up revealing 
information about $i$. Of course, one could hash (encrypt) IDs and index in the hashed space. But when the memory 
is compromised, the hash (encryption) function will get revealed and will let the adversary decode by enumerating IDs.
Often, sampling algorithms are used for providing statistical estimates, but these are vulnerable because
when the internal memory is revealed, the sampled IDs are compromised. 
\end{itemize}

To overcome the first concern, we can adopt the notion of {\em differential privacy}~\cite{DMNS}. 
Let $S'_t$ be the updates derived from $S_t$ by replacing some occurrences of some ID $j$ with occurrences of some other ID $k$. 
Informally, a randomized function $f$ is {\em differentially private} with respect to the IDs if the probability distribution induced by $f(S_t)$ on the range of $f$ is not much different from that induced by $f(S'_t)$ for any $S'_t$ as defined above, and any $t$. For the first 3 statistics listed above, using known techniques, it is straightforward to get differentially private estimates; 
for frequency moments, one can look at a related function {\em cropped frequency moment} 
$T_k(\tau)= \sum_{i \in \Uni} \min\{(\sum_{j, i_j=i} d_j)^k, \tau\}$ that bounds what is known as the 
sensitivity of the function and get differentially private estimates. 

The authors in~\cite{dwork-pan} initiated the study that addresses the second concern above. In particular, they 
defined the notion of {\em pan-privacy}. Informally, $S_t$ and $S'_t$ 
should produce very similar distributions on {\em both} internal states as well as outputs. 
Without some  ``secret state''it might seem  impossible to estimate statistics privately, but 
~\cite{dwork-pan} showed that some of the statistics above can be estimated accurately. Their main results 
were for {\em streaming algorithms}, that use space polylogarithmic in $m$ and other parameters. 
In particular, they showed pan-private streaming algorithms for rarity ratio, distinct count, cropped mean $T_1$ and a version of heavy hitters. 

We are inspired by this work~\cite{dwork-pan} and this emerging direction of pan-private algorithms~\cite{cynthia} to revisit these problems. There are some outstanding fundamental questions: 
\begin{itemize}
\item
Is there a cost to pan-privacy, that is, are there problems for which pan-privacy provably needs more resources or loss of accuracy, compared to just differential privacy?

\item
What is the impact of memory in pan-privacy? Since the memory used by the algorithm may get revealed to an
adversary, do pan-private algorithms use very small memory like in the
streaming algorithms of~\cite{dwork-pan}, or can they use large memory to better encode information about the input and get better accuracy?

\item
Technically,~\cite{dwork-pan} used samples and adapted techniques from {\em randomized response}~\cite{rr} such as distorting counters with random shifts or using two distinct distributions. In contrast, in streaming~\cite{M}, some of the most powerful algorithms use sketches that are linear projections of data along random directions. Do sketches provide improved or richer pan-private algorithms?
\end{itemize}

\paragraph{Our Contributions}
We address these questions and make the following main contributions. We focus on the basic model of pan-privacy as formulated by~\cite{dwork-pan} where memory may be breached by an adversary once unannounced to the algorithm (and later comment on the variants of the model). 
%Also let $m=|\Uni|$. 

\begin{itemize}
\item
{\em Distinct Counts}. We present a streaming algorithm that is $\eps$-pan private and outputs an estimate
$(1+\eps)\Dt \pm polylog(m)$. It directly uses sketch known before based on stable
distributions for estimating distinct counts~\cite{stablehamming}, but maintains noisy versions. In fact, this approach is powerful and our pan-private algorithms even work for {\em turnstile} streams where $d_i$'s may be negative, the first pan-private algorithms
to have this property. In contrast, best previous result for pan-private streaming estimation of distinct count 
outputs an estimate $\Dt \pm \alpha m$ with constant probability, for only nonnegative updates~\cite{dwork-pan}. Note that stable distribution based 
approach is known to yield streaming algorithms for $F_k$ for $0\leq k \leq 2$ ~\cite{Indyk}, but 
this analogy does not work to get pan-private estimate of $F_k$ (adapted as $T_k$); therefore, that it works for 
pan-private $k=0$ (which is related to distinct counts) is very interesting. 

We complement this result by showing lower bounds. Let $\mathcal{A}$
be an online (not necessarily streaming) algorithm that on input $S_t$
outputs $\Dt \pm o(\sqrt{m})$ 
%with probability at least $1 - O(m^{-2})$. 
with small constant probability.  Then $\mathcal{A}$ is not
$\epsilon$-pan private for any constant $\epsilon$. This is the
first-known lower bound for any pan-private algorithm in this
model. In fact, we develop an approach to showing lower bounds (which
may be of independent interest in the future) that takes a copy of the
memory by breaching the algorithm once, and then simulating the
algorithm with random inputs in parallel with this seed memory like
noisy decoding~\cite{dinur2003revealing}.  Our lower bound holds no
matter the memory used by $\mathcal{A}$, even if the memory is
$\Omega(m)$. Thus, $\mathcal{A}$ need not be a streaming
algorithm. Our lower bound is not like the ones in streaming
literature where the lower bound is conditioned on using small space,
or like in differentially private optimization~\cite{gupta} where one
shows structural relationship between ``near" configurations of
inputs.  We show this lower bound is essentially tight if
$\mathcal{A}$ is {\em not} streaming: we show a simple pan-private
algorithm that outputs an estimate $\Dt \pm o(\sqrt{m})$ with constant
probability and maintains $O(m)$ memory. Further, we show a lower
bound of $(1+\eps)\Dt \pm polylog(m)$, essentially tight upto additive
polylog terms with our streaming algorithm.

\item
{\em Heavy Hitters Count}.
As is standard in streaming literature, we define 
$\HH^{(t)}(k)$ as the number of IDs $i$ with $\sum_{i | i_j=i} d_j \geq F_1^{(t)}/k$. In this notation,~\cite{dwork-pan} 
approximates $\HH(k)$ within an additive error of $O(\alpha m)$ for
any constant $\alpha$. However, $m$ can far exceed $k$ which is an
upper bound on $\HH(k)$. 
%Further,~\cite{dwork-pan} weakens pan-privacy slightly to get their results on $\HH(k)$. We rexamine this problem in standard notion of pan-privacy. 

We  present a pan-private streaming algorithm that returns an estimate in 
$[(1-\eps)\HH(k) - O(\sqrt{k}), \HH(O(k^2)) + O(\sqrt{k})]$ (that is no worse than $O(k)$ approximation, upto additive errors), which is a significant improvement over~\cite{dwork-pan}. We obtain this by first observing 
that with $O(m)$ space, we can provide an estimate $\HH(k) \pm O(\sqrt{m})$, and then using this only
on the space of all buckets in the Count-Min sketch~\cite{cormode2005improved} which uses much smaller space. 
\end{itemize}

Some comments: (1)  Both of our results above are obtained using sketches, which is different from use of samples thus far~\cite{dwork-pan}. Also, we use full space versions on top of sketches to get best-known accuracies. 
(2) An interesting aspect of our pan-private algorithms is that, they deliberately use very small (polylogarithmic) space and are streaming 
even though using more space is not forbidden. (3) Our insights from above yield other upper bounds (pan-private streaming algorithms for $T_k$,
inner products of vectors etc) and lower bounds (inner products). 
(4) We are adapting pan-privacy model
from ~\cite{dwork-pan} as a given and refer the readers to that original work for motivating and 
defending the model as well as  discussion 
related to the model such as, what if a small amount of secret storage is allowed, or what if adversary is 
allowed to look at the memory multiple times or even continually and so on. For the purposes of this
paper, the basic pan-privacy model is of great interest and there are fundamental techical problems that we 
address. (5) Likewise, the specific statistics we have considered have many applications that have been 
identified over the past decade from databases to data streams, compressed sensing and beyond~\cite{M}. We 
do not elaborate on this further, instead addressing how these problems can be solved. (6) 
Finally, we have focused on counts throughout. Many of these problems have a corresponding ``list'' version
in which output comprises specific IDs. We have left it open to identify suitable pan-private versions 
of these problems. 

\paragraph{Map.} 
In Section 2, we introduce relevant definitions and notation. In
Section 3, we present our upper and lower bounds for distinct count
estimation. In Section 4, we present our upper bound for heavy hitter
count estimation.  In Section 5, we have concluding remarks with other
extensions.

%% file: sect2.tex
\section{Preliminaries}\label{sec-def}

\input{notation}

\subsection{Differential Privacy}
Dwork et al.~\cite{DMNS} introduce the concept of differential privacy which operates on a data set consisting of rows of data, where each row consists of the data of an individual. Differential privacy provides a guarantee that the probability distribution on the outputs of a mechanism is ``almost the same'', irrespective of whether an individual opts in to, or out of, the data set. Such a guarantee incentivizes participation of individuals in a database by assuring them of incurring very litle risk by such a participation. Formally,
\begin{defn}[\cite{DMNS}]
 A randomized function $f$ provides $\eps$-differential privacy if for all neighboring (differing in at most one row) data sets $D$ and $D'$, and all $Y \subseteq Range(f)$,
\[\Pr[f(D) \in Y] \leq exp(\eps) \times \Pr[f(D') \in Y].\]
\end{defn}
One mechanism that~\cite{DMNS} use to provide differential privacy is the so called ``Laplacian noise method'', which depends on the {\em global sensitivity} of a function:
\begin{defn}[~\cite{DMNS}]\label{def:sensitivity}
For $f: \mathcal{D} \rightarrow \mathbb{R}^d$, the global sensitivity of $f$ is
\[GS_f = \max_{D,D'} \norm{f(D) -f(D')}_1 \]
for all neighboring data sets $D$ and $D'$.
\end{defn}
The Laplace distribution with mean $0$ and scale parameter $b$, denoted $Lap(b)$, has density function $p(x)=\frac{1}{2b} \exp(-\abs{x}/b)$. The following theorem from~\cite{DMNS} uses the Laplace distribution to construct a differentially private mechanism:

\begin{theorem}[\cite{DMNS}]\label{thm:Laplace}
 For $f:\mathcal{D} \rightarrow \mathbb{R}$, mechanism $\mathcal{M}$ that adds independently generated noise drawn from $Lap(GS_f/\eps)$ to the output preserves $\eps$-differential privacy.
\end{theorem}

\subsection{Pan-privacy}

While differential privacy provides meaningful gurantees to mitigate the risks of an individual being identified by particpating in a data set, individuals might also be concerned about retaining similar guarantees even if the {\em internal state} is revealed, say, because of a subpoena. Mechanisms that achieve this property are called {\em pan-private}~\cite{dwork-pan}. Pan privacy guarantees a participant that his/her risk of being identified by participating in a data set is very little even if there is an external intrusion on the data. 
%To use $\eps$-differential privacy for data streams and extend it to pan-privacy, one needs to modify the concept %of neighborhood. 
%Dwork et al.~\cite{dwork-pan} distinguish between {\em event-level}  and {\em user-level} neighbors. 
%Here we present the version of ``neighborrelevant to online updates. 
Formally, consider two online updates $S= \lbrace (i_1, d_1), \ldots, (i_t, d_t) \rbrace$ and $S'=\lbrace (i'_1, d'_1), \ldots, (i'_{t'}, d'_{t'}) \rbrace$ associated with state vectors $\vect{a}$ and $\vect{a'}$ respectively. 
\begin{defn}\label{def:neigh}
$S$ and $S'$ are said to be {\em neighbors} if there exists a (multi)set of updates in $S$ indexed by $K \subseteq [t]$ that update the same ID $i \in \Uni$, and there exists a (multi)set  of updates in $S'$ indexed by $K' \subseteq [t']$ that updates some $ j(\ne i) \in \Uni$ such that $\sum_{k \in K}d_k = \sum_{ k \in K'} d'_k$ and for  all other updates in $S$ and $S'$ indexed by $Q=[t]-K$ and $Q'=[t']-K'$ respectively,
\[\forall {i \in \Uni}~ \displaystyle\sum_{ k \in Q, s.t.~ i_k=i } d_k = \displaystyle \sum_{ k \in Q', s.t.~ i'_k=i} d'_k\].
\end{defn}
Notice that in the definition above $t$ and $t'$ don't have to be
equal because we allow the $d_i$'s to be integers. The definition
ensures that two inputs are neighbors if some of the occurrences of an
ID in $S$ is replaced by some other ID in $S'$ and everything else
essentially stays the same except (a) the order may be arbitrarily
different and (b) the updates can be broken up since they are not
constrained to be $1'$'s. The neighbor relation preserves the first
frequency moment of the sequence of updates, considered to be public
information. Also, the graph induced by the neighbor relation on any
set of sequences with the same first frequency moment is connected.

\begin{defn}[User level pan-privacy\cite{dwork-pan}]\label{defn:panprivacy}
 Let $\Alg$ be an algorithm. Let $I$ denote the set of internal states of the algorithm, and let $\sigma$ the set of possible output sequences. Then algorithm $Alg$ mapping input prefixes to the range $I \times \sigma$, is {\em pan-private (against a single intrusion)\footnote{See~\cite{dwork-pan} for discussion about multiple intrusions.}} if for all sets $I' \subseteq I$ and $\sigma'\subseteq \sigma$, and for all pairs of user-level neighboring data stream prefixes $S$ and $S'$
\[\Pr[\Alg(S) \in (I', \sigma')] \leq e^\varepsilon \Pr[\Alg(S') \in (I', \sigma')] \]
where the probability spaces are over the coin flips of the algorithm $\Alg$.
\end{defn}

%% file: notation.tex
\subsection{Definitions and Notation}

We are given a universe $\Uni$, where $|\Uni| = m$. An \emph{update}
is defined as an ordered pair $(i, d) \in \Uni \times
\mathbb{Z}$. Consider a semi-infinite sequence of updates $(i_1, d_1),
(i_2, d_2), \ldots$; the \emph{input} for all our algorithms consists
of the first $t$ updates, denoted $S_t = (i_1, d_1), \ldots, (i_t,
d_t)$. The \emph{state} after $t$ updates is an $m$-dimensional vector $\va^{(t)}$, indexed by the elements in $\Uni$ (we will omit the
superscript when it is clear from the context). The elements of the
vector $\va = \va^{(t)}$, referred to as the \emph{state vector}, are defined as follows:
\begin{equation*}
  \ai = \sum_{j: i_j = i}{d_j}.
\end{equation*}
We consider two models: the \emph{cash register model} in which all
updates are positive, i.e.~$\forall j: d_j \geq 0$, and the
\emph{turnstile} model in which updates can be both positive
(\emph{inserts}), i.e.~$d_j \geq 0$, and negative (\emph{deletes}),
i.e.~$d_j < 0$. We note that the turnstile model has not been
considered in pan privacy before.

Our algorithms \emph{output} a real number which approximates one of
the following statistics on $S_t$:
\begin{itemize}
\item \emph{distinct count}: $D = \Dt = |\{i \in \Uni: \ai
  \neq 0\}|$;
\item $k$-th \emph{frequency moment}: $F_k = F^{(t)}_k = \sum_{i \in
    \Uni}{|\ai|^k}$. This coincides with the $L_k$ norm,
  $\|\vect{a}\|_k$ of the state vector $\vect{a}$ and we will use
  either terms to facilitate exposition.
\item $k$-th \emph{cropped frequency moment}: $T_k(\tau) =
  T_k^{(t)}(\tau) = \sum_{i \in \Uni}{\min\{|ai|^k, \tau\}}$.
\item \emph{cropped dot product}: Given two sequences of updates $S_t$
  and $S'_t$ with state vectors $\va$ and $\va'$, the cropped dot
  product is $(\va\cdot\va')(\tau) = \sum_{i \in \Uni}{\min\{\ai \ai',
      \tau\}}$. 
\item $k$-\emph{heavy hitters count}: $\HH(k) = \HH^{(t)}(k) = |\{i:
  |\ai| \geq F_1^{(t)}/k\}|$.
  %\begin{itemize}
  %\item When this facilitates the exposition, we treat the state array
   % $\vect{a}$ as a vector in $\mathbb{R}^m$ and refer to it as the \emph{state vector}. In this case the $k$-th
    %frequency moment $F_k$ coincides with the $L_k$ norm
   % $\|\vect{a}\|_k$. 
  %\end{itemize}
\end{itemize}

%% file: sect3.tex
\section{Distinct Count Estimation}

In this section we present upper and lower bounds for the problem of
pan-private estimation of the distinct count statistic $\Dt$. We
utilize a sketching approach based on a stable
distribution. In contrast with the sampling approach of Dwork et
al.~\cite{dwork-pan}, the sketching approach works in the more general
turnstile model and for the usual range of $\Dt$ achieves significantly
better accuracy. We present our algorithm for distinct count
estimation as evidence of the usefulness of the sketching approach for
designing pan-private algorithms. 

We compliment our upper bound with lower bounds based on noisy
decoding. Our results present the first lower bounds against
pan-private algorithms that allow a single intrusion. 

\subsection{Upper Bounds}

Consider the turnstile model where the $d_j$'s could either be positive or negative, and assume an upper bound on the absolute value of each element of the state vector: $\forall i \in \Uni, \abs{\ai} < Z$. 
We are interested in a pan-private computation of  $\Dt= \{i | \vect{a}^{(t)}[i] \not= 0\}$. Note that where the superscipts don't appear a time slice of $t$ is implicit.
Recall that the $L_p$ norm of a vector $\vect{a}$ is 
$\norm{\vect{a}}_p = (\sum_i \abs{\ai}^p)^{\frac{1}{p}}.$ 

\subsection{Prior Approach in Streaming Algorithms}

~\cite{stablehamming} show that, for sufficiently small $p~ (0<p < \epsilon/\log Z)$ 
\begin{equation} \label{eq:Dt}
\Dt \leq \sum_i \abs{\ai}^p \leq (1+\epsilon) \Dt.
\end{equation}
Hence, it suffices to estimate the $L_p$ norm of $\vect{a}$ for certain small $p$ for estimating the distinct counts. For this purpose they use what are called \emph{stable distributions}.
\paragraph{Stable distributions and their use in sketches}

A distribution $\Sd$ over $\mathbb{R}$ is said to be $p$-stable, if there exists $ p \geq 0$ such that for any $n$ real numbers $b_1, \ldots , b_m$ and i.i.d. variables $Y_1, \ldots, Y_m$ with distribution $\Sd$, the random variable $\sum_i b_iY_i$ has the same distribution as the random variable $(\sum_i \abs{b_i}^p)^{1/p} Y$, where $Y$ is a random variable with distribution $\Sd$~\cite{nolan:2010}.  
Let $X$ be a matrix of random values of dimension $m \times r$, where each entry of the matrix $X_{i,j}$, $ 1\leq i \leq m$, and $1\leq j \leq r$, is drawn independently from a random stable distribution with parameter $p$, with $p$ as small as possible. The \emph{sketch vector} $\sk$ is defined as the dot product of matrix $X^T$ with $\vect{a}$, so
\[\sk_j = \sum_{i=1}^m X_{i,j}\ai = X_j\cdot \vect{a},\]
where $X_j$ is a $m$-dimensional vector composed of the following elements:$(X_{1,j}, X_{2,j}, \ldots X_{m,j})$.

From the property of stable distributions we know that each entry of $\sk$ is distributed as $(\sum_i\abs{\ai}^p)^{1/p}X_0$, where $X_0$ is a random variable chosen from a $p$-stable distribution. The sketch is used to compute $\sum_i \abs{\ai}^p$ for $ 0<p < \eps/\log Z$, from which we can approximate $\Dt$ up to a ($1+ \eps)$ factor. By construction, any $\sk_j$ can be used to estimate $L_p^p$. ~\cite{stablehamming} obtain a good estimator for $(\sum_i \abs{\ai})^p$  by taking the median of all entries $\abs{\sk_j}^p$ over $j$: 
\begin{lemma}[\cite{stablehamming}]\label{lemma:median} With probability $1-\delta$ if $r =O(1/\epsilon^2 \cdot \log(1/\delta))$, 
 \[ (1-\epsilon)^p \median_j \abs{\sk_j}^p \leq \median \abs{X_0}^p (\sum_i \abs{a_i}^p) \leq (1+ \eps)^p \median_j \abs{\sk_j}^p \]
where  $\median \abs{X_0}^p$, is the median of absolute values (raised to the power $p$) from a $p$-stable distribution.
\end{lemma}
Using the results of Equation~\ref{eq:Dt} and Lemma~\ref{lemma:median} Cormode et al.~\cite{stablehamming} prove that:
\begin{theorem}[\cite{stablehamming}]\label{thm:approximation}
 The computation of a sketch $\sk$ of online data described by a state vector $\vect{a}$ that requires space $O(1/\epsilon^2. \log(1/\delta))$ allows an approximation of $\Dt$ within a factor of $1 \pm  \epsilon$ of the true answer with probability $1-\delta$.
\end{theorem}

\paragraph{Maintaining the sketch under updates}
As updates arrive, the sketch vector is built progressively. It is initialized to be the zero vector, and on receiving tuple $(i,d_k)$, the update is done by adding $d_k$ times $X_{i,j}$ to each entry $\sk_j,~ \forall j \in [r]$ of the sketch vector. That is,
\[ \forall j \in [r]: \sk_j \leftarrow \sk_j + d_kX_{i,j}.\]
In order to avoid percomputing and storing all the values $X_{i,j}$ Cormode et al.~\cite{stablehamming} generate the random variables $X_{i,j}$ from a stable distribution on the fly by using $i$ to seed a pseudo-random number generator $random()$. These pseudo-randomly generated numbers are then used to generate a sequence of $p$-stable distributed random variables using a (deterministic) function $stable(r_1,r_2,p)$, where $r_1$ and $r_2$ are pseudorandom variables in the range $[0 \ldots1]$ drawn from $random()$. The function is defined as follows: first define a quantity $\theta = \pi(r_1 -1/2)$. Now,
\[stable(1/2+\theta, r_2, p) = \frac{\sin p\theta}{\cos^{1/p} \theta} \left( \frac{\cos(\theta(1-p))}{-\ln r_2}\right)^{\frac{1-p}{p}}.\] 
Since each time the same seed $i$ is used, this ensures that $X_{i,j}=stable(r_1,r_2,p)$ takes the same value each time it is used. We will find this technique useful for our own purpose of precomputing the global sensitivity of a sketch in the next section.

\subsection{Pan-Private Algorithm}

To get pan-privacy, we maintain these (approximate) sketches in a differentially-private way. In particular, we maintain a noisy sketch vector where each element of the sketch vector has noise added according to the sensitivity method of~\cite{DMNS}.

\paragraph{Adding Laplacian noise to the sketches.}
The global sensitivity of a sketch $\sk_j$, ($GS_j$) from Definition~\ref{def:sensitivity} is
\[GS_j = 2\cdot Z \infinorm{X_j}. \]
Consider state vectors $\vect{a}$ and $\vect{a'}$ corresponding to two neighboring sequences of online updates $S$ and $S'$ respectively. From Definition~\ref{def:neigh} there exists some $i \in [n]$ and some $k\ne i \in [n]$, such that some occurences of $i$ in the sequence of updates in $S$ is replaced by some occurences of $k$ to get $S'$. This means that $\ai \ne a'_i$ and $a_k \ne a'_k$, and for any other $l$ not equal to $i$ or $k$, $a_l = a'_l$. So, for any neighboring $S$ and $S'$, 
\[\norm{X_j \cdot \vect{a} - X_j \cdot \vect{a'}}_1 \leq \abs{X_{i,j}a_i-X_{i,j}a'_i+X_{k,j}a_k -X_{k,j}a'_k} \leq 2 \cdot Z \infinorm{X_j}.\]
From~\cite{DMNS}, it will follow that we need to add Laplacian noise based on this sensitivity to have a differentially private description of the state at any point, which is pan-private with respect to a single intrusion.  Since the elements of $X_j$ are random quantities independent of the data, we can compute the $L_\infty$ norm of the actual vector that we end up using without compromising on privacy. However, 
the challenge is that $\infinorm{X_j}$ is not known in advance. 
The use of index $i$ to seed the pseudorandom generator and use of the pseudrandomly generated values to generate the $X_{i,j}$'s, means that this challenge can be solved by computing $\infinorm{X_j}$, $\forall j$ before the onset of 
our algorithm (shown in Algorithm~\ref{alg1}). Also to use the result of Lemma~\ref{lemma:median}, the value of $\median{\abs{X_0}^p}$, the median of absolute values from a $p$-stable distribution, needs to be computed. This is also done numerically in advance in~\cite{stablehamming}, and then the final result is scaled by this constant factor denoted as $\sfp(p)$.

Algorithm~\ref{alg1} modifies the algorithm in~\cite{stablehamming} by maintaining $\alpha$-differentially private sketches of the stream vector $\vect{a}$. 

Each sketch is initialized with a noisy value drawn from the appropriate Laplace distribution.
Formally, 
let $\skpriv= \sk_j + \eta_j$
where $\eta_j$ is a random variable drawn from a Laplacian distribution with mean 0 and scaling factor of $GS_j/\alpha$. Here $\alpha$ is the privacy parameter. Since we maintain $r$ sketches of the data, Algorithm~\ref{alg1} gives us an overall privacy of $\alpha'=\alpha r$ as per the composition theorem~\cite{DMNS}:
\begin{theorem}[\cite{DMNS}]\label{composition}
 Given mechanisms $\mathcal{M}_i~, i \in [r]$ each of which provide $\alpha_i$-differential privacy, then the overall mechanism $\mathcal{M}$ that consists of a composition of these $r$ mechanisms, provides $\left(\sum_{i \in [r]} \alpha_i\right)$-differential privacy.
\end{theorem}
\begin{algorithm}
\caption{Pan-private approximation of $\Dt$}
\label{alg1}
\begin{algorithmic}
\STATE \textbf{INPUT:}  privacy parameter $\alpha$,  $0 < p < \epsilon/Z <1$, $\infinorm{X_j} \forall  j \in [r]$ computed off-line, an $r$-dimensional noise vector $\boldsymbol{\eta}$, where $\eta_j \sim Lap(\frac{2\infinorm{X_j}Z}{\alpha}),~\sfp(p) = \median \abs{X_0}^p$ also computed off-line numerically.
\newline
\STATE initialize the $r$-dimensional sketch vector $\skpriv$, such that $\skpriv_j = \eta_j $
\FORALL{tuples $(i,d_t)$} 
\STATE initialize $random$ with $i$
\FORALL{$j=1$ to $r$}
\STATE $r1= random()$
\STATE $r2=random()$
\STATE $\skpriv_j= \skpriv_j+ d_t*stable(r1,r2, p)$
\ENDFOR
\ENDFOR
\STATE return $\Dappr=\median_j\left(\abs{\skpriv_j}^p\right)* \sfp(p)$
\end{algorithmic}
\end{algorithm}

Our main result for distinct count estimation is to prove that $\Dappr$ returned by Algorithm~\ref{alg1} provides an $\alpha'$-differentially private approximation of $\Dt:$
\begin{theorem}\label{dis-up}
With probability $1-(r+1)\delta$, Algorithm~\ref{alg1} computes an $\alpha'$-pan-private 
approximation $\Dappr $ of $\Dt$ such that 
\[ (1-\epsilon)\Dt -O\left(\poly\left(\log(m)\cdot (1+\eps) \log(\frac{1}{\delta}) \frac{1}{\alpha'}\right)\right) \leq \Dappr \leq (1+\epsilon)\Dt +O\left(\poly\left(\log(m)\cdot (1+\eps) \log(\frac{1}{\delta}) \frac{1}{\alpha'}\right)\right)\]
\end{theorem}
We will need Claim~\ref{claim:bound}, and Lemmas~\ref{lemma:folklore} and~\ref{lemma:skpriv} for this purpose:

\begin{claim}\label{claim:bound}
 For any two real numbers $x$ and $y$ and for any $p\in [0,1)$, we have
\[\abs{x}^p-\abs{y}^p \leq \abs{x+y}^p \leq \abs{x}^p+\abs{y}^p\]
\end{claim}

\begin{proof}
First, assume $x$ and $y$ are either both positive or negative.
For any $x,y \in \mathbb{R}^+$ consider functions $g_{x,y}(p) = x^p +y^p$ and $f_{x,y}(p)= (x+y)^p$. At $p=1,\forall x,y \in \mathbb{R}^+$, the two functions intersect as $x^1+y^1 =(x+y)^1$.  At $p=0$, $g_{x,y}(p) > f_{x,y}(p)$.  We want to prove that for $p \in [0.1)$,$g_{x,y}(p) \geq f_{x,y}(p)$.
For convenience, we drop the subscript $x,y$. 

WLOG assume $x> y$, then $f(p)= x^p\left(1+ \frac{y}{x}\right)^p$ and $g(p)= x^p\left(1+ (\frac{y}{x})^p\right)$
So,
\[\frac{f(p)}{g(p)} = \displaystyle \frac{\left(1+\frac{y}{x}\right)^p}{1+ \left(\frac{y}{x}\right)^p}\]
The numerator
\[\left(1+\frac{y}{x}\right)^p < 1 + \frac{y}{x},~\text{for}~ p \in [0,1), \forall~ \frac{y}{x} <1  \]
The denominator
\[1+ \left(\frac{y}{x}\right)^p > 1+ \frac{y}{x}, ~\text{for}~ p \in [0,1), \forall~ \frac{y}{x} <1 \]
\[\implies \frac{f(p)}{g(p)} < \frac{1+ \frac{y}{x}}{1+ \frac{y}{x}}=1\]
So for any $x, y \in \mathbb{R}^+$, we have $f_{x,y}(p) < g_{x,y}(p)$, for $ p \in [0,1)$.
Similarly, assume $x$ is postive and $y$ is negative, and WLOG assume $a=\abs{x}>\abs{y}=b$. Then $\abs{x+y}=a-b$, and we can similarly prove that for $a,b \in \mathbb{R}^+$, 
\[(a-b)^p \geq a^p-b^p\].

\end{proof}

\begin{lemma}\label{lemma:skpriv}
 With probability $1-\delta$, for any $j \in [r]$, with $0<p< \epsilon/Z <1$
\[\abs{\sk_j}^p -\xi \leq \abs{\skpriv_j}^p \leq \abs{\sk_j}^p + \xi\]
 where $\xi = \left(\left(\displaystyle \frac{2\cdot Z \max_j \infinorm{X_j}}{\alpha} \log(\frac{1}{\delta}\right)\right)^p$
\end{lemma}
\begin{proof}
 We have $\abs{\skpriv_j}^p = \abs{\sk_j + \eta_j}^p$. From Claim~\ref{claim:bound}, we have:
\[\abs{\sk_j}^p -\abs{\eta_j}^p \leq \abs{\skpriv_j}^p \leq \abs{\sk_j}^p + \abs{\eta_j}^p.\]
Also since $\eta_j$ is drawn from a Laplacian distribution, we know that with probability $1-\delta$, 
$\abs{\eta_j} < \frac{GS_j}{\alpha} \cdot \log(\frac{1}{\delta})\leq \frac{2 Z \max_j \infinorm{X_j}}{\alpha} \cdot\log(\frac{1}{\delta}) $.
\end{proof}

Since Algorithm~\ref{alg1} computes $\Dappr$ by taking the (scaled) median of the $\skpriv_j$'s and Lemma~\ref{lemma:median} relates the median of the $\sk_j$'s to the $\sum_i \abs{\ai}^p$, we need to bound $\median_j \skpriv_j$ in terms of $\median_j \sk_j$.

\begin{lemma}\label{lemma:folklore}
  Let $x_1, \ldots, x_r$ and $y_1, \ldots, y_r$ be two sequences of
  real numbers satisfying $\forall i: x_i - E \leq y_i \leq x_i +
  E$. Then 
  \begin{equation*}
    \median_i{x_i} - E \leq \median_i{y_i} \leq \median_i{x_i} + E.
  \end{equation*}
\end{lemma}
\begin{proof}
  Assume, WLOG, that $x_1, \ldots, x_r$ are sorted in increasing order
  and $\median_i{x_i} = x_{\lceil r/2 \rceil}$. Let $\median_i{y_i} =
  y_j$. We will prove that $y_j \geq x_{\lceil r/2 \rceil} - E$, and
  the other side of the inequality will follow by an analogous argument.

  If $j \geq \lceil r/2 \rceil$, then $y_j \geq x_j - E \geq x_{\lceil
    r/2 \rceil} - E$. Therefore, we may assume $j < \lceil r/2
  \rceil$. Because $y_j$ has rank $\lceil r/2 \rceil$ in $y_1, \ldots,
  y_r$, there exist indices $k_1, \ldots, k_{d} > j$, where $d =
  \lceil r/2 \rceil - j$, s.t.~$y_{k_1}, \ldots, y_{k_{d}} \leq
  y_j$. At least one of ${k_1}, \ldots, {k_d}$ is greater than or
  equal to $\lceil r/2 \rceil$; let the smallest such index be $\ell$.
  Then we  have,
  \begin{equation*}
    y_j \geq y_\ell \geq x_\ell - E \geq x_{\lceil r/2 \rceil} - E.
  \end{equation*}
\end{proof}

Now we prove that $\Dappr$, returned by Algorithm~\ref{alg1} gives a good approximation to $\Dt$:

\begin{lemma}\label{bi-apprx}
Algorithm~\ref{alg1} computes an $\alpha'$-pan private 
approximation of $\Dt$ using space $O(1/\epsilon^2 \log(1/\delta))$. The approximation guarantee is: 
 With probability at least $1-(r+1)\delta$ and with $\xi$ as in Lemma~\ref{lemma:skpriv},
\[(1-\epsilon)\Dt -\xi\cdot \sfp(p) \leq \Dappr \leq (1+\epsilon)\Dt+ \xi \cdot \sfp(p)\]
where $r=O(1/\eps^2 \cdot \log 1/\delta)$.
\end{lemma}

\begin{proof}
Since each sketch $\skpriv_j$ is $\alpha$-differentially private according to the sensitivity method of~\cite{DMNS}, and we have $r$ such sketches, the over all privacy of the Algorithm is $\alpha r=\alpha'$. Each sketch is a differentially private description of the state and hence the algorithm achieves $\alpha'$ pan-privacy. Nowe we prove the approximation guarantee:

 We have $\Dappr=\median_j \skpriv_j \cdot \sfp(p)$. Using Lemma~\ref{lemma:skpriv} we have with probability at least $1-r\cdot \delta~\forall j$ simulatenously:
\[\abs{\sk_j}^p - \xi \leq \abs{\skpriv_j}^p \leq \abs{\sk_t}^p + \xi. \]

From Lemma~\ref{lemma:folklore}, we have with probability at least $1-r \delta$:

So we have with probability $1-r\delta$
\[\median_j \abs{\sk_j}^p -\xi \leq \median_j \abs{\skpriv_j}^p \leq \median_j \abs{\sk_j}^p + \xi.\]
Using Lemma~\ref{lemma:median} and Eqaution~\ref{eq:Dt} and noting that $\alpha' = \alpha r$, the result follows.
\end{proof}

Since $p < \epsilon/\log Z<1$, and $r$, the number of sketches is polylogarithmic in $m$, and $\sfp(p)$ is a constant, from Lemma~\ref{bi-apprx}, we have:

\begin{theorem}\label{dis-up}
With probability $1-(r+1)\delta$, Algorithm~\ref{alg1} computes an $\alpha'$-pan-private 
approximation $\Dappr $ of $\Dt$ such that 
\[ (1-\epsilon)\Dt -O\left(\poly\left(\log(m)\cdot (1+\eps) \log(\frac{1}{\delta}) \frac{1}{\alpha'}\right)\right) \leq \Dappr \leq (1+\epsilon)\Dt +O\left(\poly\left(\log(m)\cdot (1+\eps) \log(\frac{1}{\delta}) \frac{1}{\alpha'}\right)\right)\]
\end{theorem}
\begin{proof}
 We have,
\[\xi.\sfp(p) = O\left(\left(\frac{1}{\alpha'} r \cdot 2 Z\max_j \infinorm{X_j} \log(\frac{1}{\delta})\right)^p\right)\]
and $Z^p < e^\epsilon$, which for small $\epsilon$ is less than $(1+ \eps)$. Since $r$ is $\polylog$ in $m$ and $\max_j \infinorm{X_j}$ is a constant, the result follows. 

\end{proof}

In fact, this algorithm is a streaming algorithm since it stores polylogarithmic in $m$ space and takes time
polylogarithmic in $m$ per new update. Technically, it works in the {\em turnstile} model since $d_j$ may be 
positive or negative, the first such pan-private streaming
algorithm~\cite{M}.

The best previous result for pan-private distinct count estimation is
due to Dwork et al.~\cite{dwork-pan}. Their algorithm outputs and
estimate in $[\Dt - \alpha'm, \Dt + \alpha'm]$ with probability
$1-\delta$ for any constant $\alpha$ and $\delta$. By extending their
techniques and running their algorithm in full space, we can get an
estimate in $[\Dt - O(\sqrt{m}), \Dt + O(\sqrt{m})]$ with constant
probability (see Section~\ref{HH}). Our sketching algorithm achieves a
significantly smaller error whenever $\Dt = o(\sqrt{m})$; we note that
in practice the distinct counts statistic is usally much smaller than
the size of the universe.

\subsection{Lower bounds}
\label{LB}

Next we present lower bounds against pan-private algorithms that allow
a single intrusion. These are the first such lower bounds in the
literature and may be of independent interest.

We show that if only an additive approximation is allowed, the full
space extension of Dwork et al.'s algorithm for distinct count
estimation, as presented in Section~\ref{HH}, is optimal. Thus, the
multiplicative approximation factor in the analysis of our sketching
distinct counts algorithm is necessary. Furthermore, by proving a new
noisy decoding theorem, we show that our sketching algorithm gives an
almost optimal bi-approximation guarantee. Interestingly, our lower
bounds make no assumptions on the space complexity of the algorithm,
and yet the (almost) optimal algorithm happens to use polylogarithmic
space. 

\paragraph{Dinur-Nissim Style Decoding}
Our lower bounds utilize a decoding algorithm of the style introduced
in a privacy context by Dinur and
Nissim~\cite{dinur2003revealing}. Informally, we argue that the
(private) state of an accurate pan private algorithm can be used to
recover the majority of the algorithm's input. First, we introduce the
decoding results we will use. 

\begin{theorem}[\cite{dinur2003revealing}]
  \label{thm:din-nis}
  Let $\vx \in \{0, 1\}^n$. For any $\epsilon$ and $n \geq n_\epsilon$,
  the following holds. Given $O(n \log^2 n)$ random strings $\vq_1,
  \ldots, \vq_t \in_R \{0, 1\}^n$, and approximate answers $\tilde{\va}_1,
  \ldots, \tilde{\va}_t$ s.t.~$\forall i \in [t]: |\vx\cdot \vq_i -
  \tilde{\va}_i| = o(\sqrt{n})$, there exists an algorithm that outputs
  a string $\tilde{\vx} \in \{0, 1\}^n$ and except with negligible
  probability $||\vx - \tilde{\vx}||_0 \leq \epsilon n$.
\end{theorem}

In follow up work, ~\cite{dwork2007price} strengthened
the above and showed that decoding is possible even
when a constant fraction of the queries are inaccurate.

\begin{theorem}[\cite{dwork2007price}]
  \label{thm:lp-dec}
  Given $\rho < \rho^*$, where $\rho*$ is a constant approximately
  equal to $0.239$, there exists a constant $\epsilon$ s.t.~the
  following holds. Let $\vx \in \{0, 1\}^n$. There exists a matrix $A
  \in \{-1, 1\}^{n \times m}$ for some $m = O(n)$ and an efficient
  algorithm $\mathcal{A}$, s.t.~on input $\tilde{b} \in \mathbb{N}^m$,
  satisfying $|\{i: |(A\vx - \tilde{b})_i| > \alpha\}| \leq \rho$,
  $\mathcal{A}$ outputs $\tilde{\vx} \in \{0, 1\}^n$ and with
  probability $1-e^{-O(m)}$, $||\vx - \tilde{\vx}||_0 \leq \epsilon \alpha^2$
\end{theorem}

Next we will prove a result that is similar to Dinur and Nissim's but
uses ``union queries'' as opposed to dot product queries. 

\begin{theorem}
  \label{thm:union-dec}
  Let $\vx \in \{0, 1\}^n$, $||\vx||_0 \leq C\log^c n$ for some constants
  $c$ and $C$. For any $\epsilon$ and $n \geq n_\epsilon$ the following
  statement holds. There exists $n^{O(log^c n)}$ binary strings $\vq_1,
  \ldots, \vq_t \in \{0, 1\}^n$ and an algorithm $\mathcal{A}$ such that
  given answers $\tilde{\va}_1, \ldots, \tilde{\va}_1$ satisfying
  $$
  \forall i: (1-\alpha_1)||\vx + \vq_i||_0 - \alpha_2 \leq \tilde{\va_i}
  \leq (1+\alpha_1)||\vx + \vq_i||_0 + \alpha_2
  $$
  for $\alpha_2 = o(\log^c n)$, $\mathcal{A}$ outputs $\tilde{\vx}$ with
  $||\vx - \tilde{\vx}||_0 \leq \frac{16(\alpha_1 + \epsilon)}{1 -
    \alpha_1} C \log^c n$.
 \end{theorem}
 \begin{proof}
   Let $L$ be an upper bound on $||\vx||_0$, i.e.~$L = C \log^c n$. The
   set of queries is $\vq_0 = (0, \ldots, 0)$, and $\vq_1, \ldots, \vq_t$
   are the indicator vectors of all subsets of $[n]$ of size at most
   $L$. The algorithm outputs any string $\tilde{\vx}$
   s.t.~$||\tilde{\vx}||_0 \leq L$ and $\tilde{\vx}$ satisfies all the
   following constraints:
   $$
   \forall i: (1-\alpha_1)||\tilde{\vx} + \vq_i||_0 - \alpha_2 \leq \tilde{a_i}
   \leq (1+\alpha_1)||\tilde{\vx} + \vq_i||_0 + \alpha_2
   $$
   Clearly the algorithm terminates, as at least one string, i.e.~$\vx$
   satisfies all constraints. Choose $\epsilon$ so that $\alpha_2 \leq
   \epsilon L$. Next we argue that if $||\vx - \tilde{\vx}||_0 >
   \frac{16(\alpha_1 +\epsilon)}{1 - \alpha_1}L$, at least one of the
   above constraints is violated.

   We will consider several cases. Let $b = \frac{2(\alpha_1 +
     \epsilon)}{1-\alpha_1}L$. Assume first that $||\vx||_0 -
   ||\tilde{\vx}||_0 > b$. Then,
   \begin{align*}
   \tilde{a_0} &\geq (1-\alpha_1)||\vx||_0 - \alpha_2\\
   &> (1-\alpha_1)(||\tilde{\vx}||_0 + b) - \alpha_2\\
   &\geq (1 + \alpha_1)||\tilde{\vx}||_0 + \alpha_2 -
   (2\alpha_1||\tilde{\vx}||_0 + 2\alpha_2 - (1-\alpha_1) b)\\
   &\geq (1 + \alpha_1)||\tilde{\vx}||_0 + \alpha_2 -
   (2(\alpha_1 + \epsilon)L - (1 - \alpha_1) b)\\
   &\geq (1 + \alpha_1)||\tilde{\vx}||_0 + \alpha_2.
   \end{align*}
   We have shown that a constraint is violated by $\tilde{\vx}$ in this
   case. The case  $||\tilde{\vx}||_0  - ||\vx||_0 > b$ is argued
   analogously. 

   Finally, assume that $||\tilde{\vx}||_0 - ||\vx||_0 \in [-b, b]$. Let
   $\vq'$ be the indicator vector of the set $\{i: \vx_i = 0, \tilde{\vx}_i
   = 1\}$, and, similarly, let $\vq''$ be the indicator vector of the
   set $\{i: \vx_i = 1, \tilde{\vx}_i = 0\}$. Since, by assumption $||\vx -
   \tilde{\vx}||_0 = ||\vq'||_0 + ||\vq''||_0 > \frac{16(\alpha_1
     +\epsilon)}{1-\alpha_1}L$, it follows that $\max(||\vq'||_0,
   ||\vq''||_0) > \frac{8(\alpha_1 + \epsilon)}{1- \alpha_1}L$. Assume,
   without loss of generality, that $||\vq'||_0 > 8\frac{8(\alpha_1 +
     \epsilon)}{1-\alpha_1}$. We have the following identities:
   \begin{align*}
     ||\vq' + \vx||_0 &= ||\vx||_0 + ||\vq'||_0\\
     ||\vq' + \tilde{\vx}||_0 &\leq ||\tilde{\vx}||_0 + ||\vq'||_0 - \frac{4(\alpha_1 +
       \epsilon)}{1- \alpha_1}L\\
     ||\vq' + \vx||_0 - ||\vq' + \tilde{\vx}||_0 &\geq \frac{4(\alpha_1 +
       \epsilon)}{1- \alpha_1}L - b
   \end{align*}
   Let $\tilde{a}$ be the approximate answer to the query $||\vq' +
   \vx||_0$. 
   \begin{align*}
     \tilde{a} &\geq (1 - \alpha_1)||\vx+ \vq'||_0 - \alpha 2\\
     &> (1 - \alpha_1)(||\tilde{\vx} + \vq'||_0 + 2b - \frac{8(\alpha_1 +
     \epsilon)}{1-\alpha_1}L) - \alpha_2\\
     &= (1 + \alpha_1)||\tilde{\vx}+ \vq'||_0 + \alpha_2
     -(2\alpha_1||\tilde{\vx} + \vq'||_0 + 2\alpha_2 + 2(1 - \alpha_1)b -
     8(\alpha_1 + \epsilon)L)\\ 
     &\geq (1 + \alpha_1)||\tilde{\vx} + \vq'||_0 - \alpha_2\
     - (2\alpha_1L + 2\epsilon L + 4(\alpha_1 + \epsilon)L -
     8(\alpha_1 + \epsilon)L)\\ 
     &\geq  (1 + \alpha_1)||\tilde{\vx} + \vq'||_0 - \alpha_2
   \end{align*}
   Therefore, the constraint is violated and this completes the proof.
 \end{proof}
 
\paragraph{Lower Bounds from Noisy Decoding}
We introduce our approach to proving lower bounds for pan-private
algorithms using the most direct argument first: a lower bound against
dot product. We introduce the problem first.

\begin{problem}

  Input is a sequence of updates $S_t$ followed by a sequence $S'_t$.

  Output: Let $\vect{a}$ be the state of sequence $S_t$, and let
  $\vect{a'}$ be the state of $S'_t$. Output $\vect{a} \cdot \vect{a'}
  \pm \alpha = \sum_{i \in \Uni}{\ai \ai'} \pm \alpha, $ where
  $\alpha$ is an approximation factor.
\end{problem}

\begin{theorem}
  \label{thm:dot-prod-lb}
  Let $\mathcal{A}$ be a streaming algorithm that on input streams
  $S_t$, $S'_t$ outputs $\vect{a}\cdot\vect{a'} \pm o(\sqrt{m})$
  with probability at least $1 - O(m^{-2})$. Then $\mathcal{A}$ is not
  $\epsilon$-pan private for any constant $\epsilon$.
\end{theorem}
\begin{proof}
  Fix a stream $S_t$ s.t.~$\forall i\in \Uni: \ai \in \{0, 1\}$. Let
  the internal state of the algorithm $\mathcal{A}$ after processing
  $S_t$ be $X$. By the definition of pan privacy, $I$ is
  $\epsilon$-differentially private with respect to $S_t$. Fix some
  constants $\delta$ and $\eta$. We will show that for all large
  enough $m$, any algorithm $\mathcal{Q}$ that takes as input $X$ and
  a stream $S'_t$ and outputs $\vect{a}\cdot\vect{a}' \pm o(\sqrt{m})$
  with probability at least $1 - O(m^{-2})$ can be used to recover
  $\ai$ exactly for all but an $\eta$ fraction of $i \in \Uni$ with
  probability $1-\delta$. Therefore, the existence of such an
  algorithm $\mathcal{Q}$ implies that $X$ cannot be
  $\epsilon$-differentially private for any fixed $\epsilon$. Indeed,
  assume for the sake of contradiction that an algorithm with the
  given properties exists and $X$ is $\epsilon$-differentially
  private. Since $\mathcal{Q}$ depends only on $X$ and not on $S_t$,
  the output of $\mathcal{Q}$ is also $\epsilon$-differentially
  private. This is a contradiction, since the output of $\mathcal{Q}$
  can be used to guess a bit of the binary vector $\vect{a}$
  accurately with probability at least $(1- 2\delta - \eta)$, where
  $\delta$ and $\eta$ can be chosen arbitrarily small.

  To finish the proof we show that an algorithm $\mathcal{Q}$ with the
  specified properties can be used to recover all but an $\eta$
  fraction of $\vect{a}$ with probability $1 - \delta$. To see this,
  observe that $\mathcal{Q}$ can be used to answer queries
  $\vect{a}\cdot\vq$ for any arbitrary $\vq$ to within $o(\sqrt{m})$
  additive error. In particular, to answer queries $\vect{a}\cdot
  \vq_1, \ldots, \vect{a}\cdot \vq_r$, run $\mathcal{Q}(X, S_t^{(1)}),
  \ldots, \mathcal{Q}(X, S_t^{(r)})$ in parallel, where $S_t^{(i)}$ is
  a stream with state $\vq_i$. If $r = o(n^2)$, then, by the union
  bound, with probability $1-\delta$ for any constant $\delta$,
  $\mathcal{Q}(X, S_t^{(i)}) = \vect{a} \cdot \vq_i \pm
  o(\sqrt{m})$. By Theorem~\ref{thm:din-nis}, there exists an
  algorithm that, given the output of $\mathcal{Q}(X, S_t^{(1)}),
  \ldots, \mathcal{Q}(X, S_t^{(r)})$, outputs $\tilde{\vect{a}}$
  s.t.~except with negligible probability $\tilde{\vect{a}}$ agrees
  with $\vect{a}$ on all but $\eta$ fraction of the coordinates.
\end{proof}

Notice that the lower bound relies on the fact that the updates for
$S_t$ arrive before any of the updates of $S'_t$. This restriction can
be relaxed. In general, we get a lower bound of $\Omega(\sqrt{m_0})$
for the additive error, where $m_0$ is the largest number of items in
$S$ that are updated before any of the corresponding items in
$S'$. The lower bound is interesting whenever the updates to the two
sequences of updates are not ``synchronized'', i.e.~$(i, d) \in S_t$
$(i, d') \in S'_t$ for the same $i \in \Uni$ are allowed to arrive at
different time steps.

Recall that the distinct count for $S_t$ is
$\Dt$. We have the following corollary.
\begin{cor}
  Let $\mathcal{A}$ be an online algorithm that on input 
  $S_t$ outputs $\Dt \pm o(\sqrt{m})$ with probability at least $1 -
  O(m^{-2})$. Then $\mathcal{A}$ is not $\epsilon$-pan private for any
  constant $\epsilon$. 
\end{cor}
\begin{proof}
  Notice that the proof of Theorem~\ref{thm:dot-prod-lb} goes through
  if we restrict the instances to be binary, i.e.~if we require that
  $\forall i \in [m]: \vect{a}, \vect{a}' \in \{0, 1\}$. The corollary
  follows by a reduction from this restricted dot-product problem to
  the distinct elements problem. Given binary streams $S'_t$, $S''_t$, let $S_t
  = (S'_t, S''_t)$ be their concatenation. By a simple application of
  inclusion-exclusion, $\Dt = \Dt(S') + \Dt(S'') -
  \vect{a}\cdot\vect{a}'$. Therefore, an $\epsilon$-pan private algorithm
  for $\Dt$ that achieves additive approximation $\alpha$ with
  probability $1-\delta$ implies a $3\epsilon$-pan private algorithm
  for dot product on binary instances that achieves additive
  approximation $3\alpha$ with probability $1-3\delta$.
\end{proof}

The next two theorems follow by arguments identical to the one used to
prove Theorem~\ref{thm:dot-prod-lb}, but using, respectively,
Theorem~\ref{thm:lp-dec} and Theorem~\ref{thm:union-dec} in place of
Theorem~\ref{thm:din-nis}.

\begin{theorem}
  Let $\mathcal{A}$ be an online algorithm that on inputs
  $S_t$, $S'_t$ outputs $\vect{a}\cdot\vect{a}' \pm o(\sqrt{m})$
  with probability at least $1 - \delta$. If $\delta < \rho^*/2(1+\eta)$
  for any $\eta$, then $\mathcal{A}$ is not $\epsilon$-pan private for
  any constant $\epsilon$.
\end{theorem}
\begin{proof}
  The proof is analogous to the proof of
  Theorem~\ref{thm:dot-prod-lb}. Note first that the $\{-1, 1\}$
  queries of Theorem~\ref{thm:lp-dec} can be simulated as the
  difference of two $\{0, 1\}$ queries, which gives $o(\sqrt{m})$
  additive error with probability at most $1-2\delta$. In order to
  apply Theorem~\ref{thm:lp-dec}, we need to guarantee that at most
  $\rho < \rho^*$ fraction of the queries answered by $\mathcal{Q}$
  have error $\Omega(\sqrt{m})$. Call such queries
  \emph{inaccurate}. In expectation there are at most $2\delta$
  inaccurate queries. Since the statement of Theorem~\ref{thm:lp-dec}
  holds when the queries are independent, an application of a Chernoff
  bound with a large enough number of queries shows that except with
  negligible probability there are at most $\rho^*$ inaccurate
  queries. After applying Theorem~\ref{thm:lp-dec} the proof can be
  finished analogously to the proof of Theorem~\ref{thm:dot-prod-lb}.
\end{proof}

\begin{cor}
  Let $\mathcal{A}$ be an online algorithm that on input 
  $S$ outputs $\Dt(S) \pm o(\sqrt{m})$ with probability at least $1 -
  \delta$. If $\delta < \rho^*/6(1+\eta)$, then $\mathcal{A}$ is not
  $\epsilon$-pan private for any constant $\epsilon$.
\end{cor}

This corollary implies the optimality of the full-space distinct
counts estimation algorithm presented in Section~\ref{HH} \emph{when
  only additive approximations are allowed}. 

Using similar arguments, we can show the following (proof omitted).
\begin{theorem}
  Let $\mathcal{A}$ be a streaming algorithm that on input a stream
  $S_t$ and any constant $\alpha$ outputs $(1 \pm \alpha) \Dt \pm
  o(\log^c m)$ with probability at least $1 - n^{-\Omega(\log^c
    m)}$. Then $\mathcal{A}$ is not $\epsilon$-pan private for any
  constant $\epsilon$.
\end{theorem}

The theorem establishes that when an arbitrarily small multiplicative
approximation factor is allowed, an additive polylogarithmic error is
unavoidable for the problem of estimating distinct counts. Thus, up to
the exact order of the polylogarithmic additive factor, our sketching
algorithm for distinct count estimation is optimal. 

%% file: section4.tex
\section{Heavy Hitters}
\label{HH}

We provide further evidence for the usefulness of sketching for
pan-private algorithms by presenting an improved algorithm for the
Heavy Hitters problem. As a tool we use a variant of the cropped mean
estimator from~\cite{dwork-pan}, but we combine it with a sketching
approach in the style of CM sketches~\cite{cormode2005improved},
instead of the sampling approach used in \cite{dwork-pan}. This will
allow us to significantly reduce the approximation error by reducing
the universe size while approximately preserving the number of heavy
items. Once again, we observe that polylogarithmic space complexity is
a by-product of the improved approximation ratio. 

\subsection{Full Space Cropped Sum}
We begin with an analysis of the cropped sum estimator in full
space. For completeness we describe the estimator. We will approximate
$T_k^{(t)}(\tau)$ for a universe $\Uni$ and a sequence of updates $S_t.$

% Let $\csum(\Uni,
% S, t)$ for a sequence of updates $S$ be defined as follows
% $
% \csum(\Uni, S, t) := \sum_{i \in \Uni}{\min(\sum_{j: i_j = i}{d_t}, t)}.
% $

Let $\mathcal{D}_0$ be the uniform distribution over $\{0, 1\}$ and
$\mathcal{D}_1$ be the distribution that assigns probability $1/2 +
\epsilon/4$ to 1 and the remaining probability to 0. We compute an
estimate $\tilde{T}_k(\tau)$ of $T_k(\tau)$ as follows:
\begin{itemize}
\item For each $j \in \Uni$, initialize a counter $c_j \in_R \{0,
  \ldots \tau-1\}$, a bit $b_j \sim \mathcal{D}_0$
\item When item $j$ arrives on the stream, increment the counter
  $c_{j} \pmod{\tau}$. If $c_{j} = 0$ pick $b_{j}$ from
  $\mathcal{D}_1$.
\item At query time, compute $o := |\{j: b_j = 1\}|$, and output
  $\tilde{T}(\tau) = (o - |\Uni|/2)\frac{4\tau}{\epsilon}$.
\end{itemize}
Note that this algorithm is simply an instantiation of the cropped
mean estimator from \cite{dwork-pan} in full space. Keeping counters
for each element  allows us to guarantee smaller
additive error in terms of $m$.

\begin{lemma}\label{lm:csum-full}
  The estimator $\tilde{T}_k(\tau)$ is $\epsilon$-differentially
  private. Moreover, with probability $1 - 2e^{-2 \alpha}$,
  $$
  |T_k(\tau) - \tilde{T}_k(\tau)| \leq \frac{4\alpha
    t\sqrt{|\Uni|}}{\epsilon} .
  $$
\end{lemma}
\begin{proof}
  The privacy analysis is identical to the privacy analysis of the
  cropped mean estimator in \cite{dwork-pan}. 
  By the analysis in \cite{dwork-pan}, $\expect[o] = |\Uni|/2 +
  \epsilon T_k(\tau)/4\tau$, and, therefore,
  $\expect[\tilde{T}_k(\tau)] = T_k(\tau)$. By Hoeffding's
  bound,
  $\prob[|o - \expect[o]| \geq \alpha \sqrt{|\Uni|}] \leq 2e^{-2\alpha}.$
  The lemma follows.
\end{proof}

Note that setting the cropping parameter $t$ to 1 gives an estimate
of the distinct count $D$ in full space with $O(\sqrt{m})$ additive
error.

\subsection{$\HH$ Algorithm}
The limiting factor in the cropped sum estimator is $m$. Even though we allow full space to
the algorithm and it achieves pan-privacy, the approximation
guarantees involve an additive factor in $m$
which is large.  The key step in our algorithm is to
project the input $S$ onto $S'$ over a much smaller
universe, so that $S'$ has approximately the same  $k$-heavy
hitters count. In fact, we are able to reduce the universe size to a
constant that depends only on $k$ and the desired approximation
guarantee. The reduced universe size directly implies a more accurate
cropped sum estimate and, hence, a more accurate estimate of the
number of $k$-heavy hitters. Next we present our algorithm.

Assume the value $F_1 = F_1^{(t_0)}$, where $t_0$ is the time step when the
algorithm will be queried, is known ahead of time. Assume also we have
oracle access to a random function $f:[m] \rightarrow [h]$ (these
assumptions will be removed in Section~\ref{ext}. Given a
sequence of updates $S$, let $f(S)$ be the sequence $(f(i_1), d_1),
\ldots, (f(i_t), d_t)$, and let $T_k(\tau|f)$ and
$\tilde{T}_k(\tau|f)$ be, respectively, $T_k(\tau)$ and
$\tilde{T}_k(\tau)$ computed on the stream $f(S)$. Note that $f(S)$ is
a stream over the universe $[h]$ and can easily be simulated online
given the oracle for $f$.

\begin{itemize}
\item Choose a random function $f:\Uni \rightarrow [h]$.  Compute $x_1
  = \tilde{T}_k(F_1/k|f)$ and $x_2 = \tilde{T}_k(F_1/ck|f)$. Output
  $$
  \tilde{\HH}(k) := (x_1 - x_2)\left(\frac{F_1}{k} - \frac{F_1}{ck}\right)^{-1}
  $$
\end{itemize}

The above algorithm will be accurate provided that the function $f$
approximately preserves the number of heavy hitters. In the next
section we show that a random $f$ satisfies this condition with high
probability.

\subsection{Reducing the Universe Size}

Remember that we denote $\ai = \sum_{j: i_j = i}{d_j}$. 
\begin{lemma}
  \label{lm:k-apx-sep}
  Let $f: \Uni \rightarrow [h]$ be a random function. Also, let
  $\tilde{k} = |\{j: \exists i \in h^{-1}(j)\text{ s.t. }\ai
  \geq t/k\}|$.  With probability $1-\delta$,
  $$
  \frac{\tilde{k}}{HH(k)} \geq 1 - \frac{k}{\delta h}.
  $$
\end{lemma}
\begin{proof}
  Let the indicator random variable $I_j$ be equal to 1 iff\; $\forall i
  \in h^{-1}(j):\ai < t/k$. The expected value of
  $I_j$ for any $j$ is as follows:
  $$
  \expect[I_j] = \left(1 - \frac{1}{h}\right)^{HH(k)} \leq
  \exp(-HH(k)/h).
  $$
  Denote, for convenience, $r := h/HH(k)$. We can write $\tilde{k}$ in
  terms of $I_j$:
  \begin{align*}
    \expect[\tilde{k}] &= \sum_{j \in [h]}{(1 - I_j)} \geq h(1 -
    e^{-1/r})\\
    \expect\left[\frac{\tilde{k}}{HH(k)}\right] &\geq r(1 -
      e^{-1/r}).
  \end{align*}
  Using the inequality $e^x \geq 1 + x + x^2$ (valid for $x \in
  [-1, 1]$), we simplify to
  $$
  \expect\left[\frac{\tilde{k}}{HH(k)}\right] \geq r(1 - 1 + \frac{1}{r} -
  \frac{1}{r^2}) = 1 - \frac{1}{r}
  $$
  We can apply Markov's inequality to the random variable $(HH(k) -
  \tilde{k})/HH(k) > 0$. Therefore, with probability $1 - \delta$,
  $$
  \frac{\tilde{k}}{HH(k)} \geq   1 - \frac{1}{\delta r} \geq
  1- \frac{k}{\delta h}.
  $$
\end{proof}

% \begin{lemma}
%   \label{lm:k-sep}
%   Let $f: [m] \rightarrow [h]$ be a pairwise-independent hash
%   function. If $h > k^2/2\delta$, with probability $1-\delta$, for any $j \in
%   [h]$ $|\{i: h(i) = j, n_i \geq n/k\}| \leq 1$.
% \end{lemma}
% \begin{proof}
%   This is a standard result. There are at most $k$ items $i$ s.t. $n_i
%   \geq n/k$, therefore there are at most ${k \choose 2}$ pairs of such
%   items. Two of them collide with probability $1/h$, and, by a union
%   bound we have 
%   $$
%   \prob[\exists i, j: n_i \geq n/k, n_j \geq n/k,f(i) = f(j)] \leq
%   \sum_{\begin{subarray}{c} i: n_i \geq n/k\\ j:n_j \geq
%       n/k\end{subarray}}{\prob[f(i) = f(j)]} < \delta
%   $$
% \end{proof}

In the next lemma we show that we can project the universe onto a
significantly smaller universe without creating ``new'' heavy
hitters. 

\begin{lemma}
  \label{lm:ck2k}
  Let $A \subseteq \Uni$ be set of items s.t. $\forall i \in A: \ai
  \leq F_1\delta/2k^2$. Also, let $f:\Uni \rightarrow [h]$ be a
  pairwise-independent hash function. There exists an $h_0 =
  \Theta(k)$, s.t.~for any $h \geq h_0$ with probability at least
  $1-\delta$
  $$
  \forall j \in [h]: \sum_{i \in A \cap f^{-1}(j)}{\ai} \leq F_1/k.
  $$
\end{lemma}
\begin{proof}
  Let $N_j = \sum_{i \in A \cap f^{-1}(j)}{\ai}$, i.e.~$N_j$ is the total
  frequency of the items mapped to $j$ by $f$. It's easy to see that
  $\expect[N_j] = \frac{1}{h}\sum_{i \in A}{\ai} \leq F_1/h$. Let's
  analyze the variance. Let $X_{ij}$ be the indicator variable for the
  event $\{i \in B_j\}$. By pairwise independence, $\Var(N_j) = 
  \sum_{i\in A}{\Var(X_{ij}\ai)}$.
  \begin{align*}
  \Var(X_{ij}\ai) &= \expect[(X_{ij}\ai)^2] - \expect[(X_{ij}\ai)]^2 \\
  &= \ai^2\left(\frac{1}{h} - \frac{1}{h^2}\right).
  \end{align*}
  Therefore, $\Var(N_j) = \left(\frac{1}{h} -
    \frac{1}{h^2}\right)\sum_{i \in A}{\ai^2}$. We will denote
  $\sum_{i \in A}{\ai^2}$ as $F_2(A)$. 

  \begin{fact} \label{fct:f2-bound} If $\sum_{i \in A}{\ai} \leq F_1$
    and $\forall i \in A: \ai \leq pF_1$ for some $p \in [0, 1]$,
    $F_2(A) = \sum_{i \in A}{\ai^2} \leq p(F_1)^2$.
  \end{fact}
  \begin{proof}
    Let $\vect{a}$ be a vector that maximizes $F_2(A)$. We may assume
    without loss of generality that $\sum_{i \in A}{\ai} = F_1$. Then
    either 0 or at least two coordinates in $\vect{a}$ can be in the
    open interval $(0, pF_1)$. We claim that there exists a maximum
    $\vect{a}$ s.t.~all coordinates are equal to either 0 or
    $pF_1$. Assume, for contradiction, that there exist $i$ and $i'$
    in $A$ s.t.~$0 < \ai < pF_1$ and $0 < \vect{a}_{i'} < pF_1$. Let
    $\ai \geq \st{i'}$. Then changing $\ai$ to $\ai + 1$ and
    $\st{i'}$ to $\st{i'} - 1$ strictly increases $F_2(A)$
    which is a contradiction.
  \end{proof}

  By Fact~\ref{fct:f2-bound}, $F_2(A) \leq (F_1)^2\delta/2k^2$. Set $h
  \geq h_0 = (\sqrt{2} + 2)k$. By the one-sided Chebyshev inequality,
  \begin{align*}
    \prob[N_j \geq \frac{F_1}{k}] &\leq \frac{1}{1 + \left(\frac{F_1}{k} -
        \frac{F_1}{h}\right)^2/\left(\frac{1}{h} -
        \frac{1}{h^2}\right)F_2(A)}\\
    &< \frac{F_2(A)}{h\left(\frac{F_1}{k} -
        \frac{F_1}{h}\right)^2}
    \leq \frac{\delta}{2hk^2\left(\frac{1}{k} -
        \frac{1}{h}\right)^2}
    = \frac{\delta}{2h\left(1 - \frac{1}{\sqrt{2} + 2}\right)^2} =
    \frac{\delta}{h}. 
  \end{align*}
  The lemma follows by a union bound.
\end{proof}

We are now ready to analyze $\tilde{\HH}(k)$. The following theorem
shows that $\tilde{\HH}(k)$ is in the range $[(1-\beta)\HH(k) -
O(\sqrt{k}), \HH(O(k^2)) + O(\sqrt{k})]$ with constant probability.
\begin{theorem}
  $\tilde{\HH}(k)$ can be computed while satisfying $2\epsilon$-pan
  privacy. Moreover, if $h \geq \max\{k/\beta\delta,
  (\sqrt{2}+2)ck\}$, then with probability $1-2\delta -
  4\exp(-\alpha)$
  $$
  (1 - \beta)\HH(k) - \frac{4(c+1)\alpha\sqrt{h}}{(c-1)\epsilon} \leq
  \tilde{\HH}(k) \leq \HH(2c^2k^2/\delta) +
  \frac{4(c+1)\alpha\sqrt{h}}{(c-1)\epsilon}.
  $$
\end{theorem}
\begin{proof}
  The privacy guarantee follows by the $\epsilon$-pan privacy of the
  cropped sum estimators and the composition theorem of Dwork et
  al.~\cite{DMNS}. Next we analyze utility.
  
  Computing cropped $F_1$ at two levels of the cropping parameter gives
  us an approximation of the number of heavy hitters:
  \begin{align*}
    T_1(F_1/k) - T_1(F_1/ck) &= \sum_{j: N_j
      \geq F_1/ck}{\min(N_j, F_1/k) - F_1/ck}\\
    &= \sum_{j: N_j \geq F_1/k}{(F_1/k - F_1/ck)} + \sum_{j: F_1/ck \geq N_j
      \geq F_1/k}{(N_j - F_1/ck)}
  \end{align*}
  It immediately follows that $|\{j: N_j \geq F_1/k\}| < \expect[\tilde{\HH}(k)] \leq
  |\{j: N_j \geq F_1/ck\}|$. By Lemma~\ref{lm:k-apx-sep}, $|\{j: N_j
  \geq F_1/k\}| \geq (1 - \beta)\HH(k)$ except with probability
  $\delta$. We can apply Lemma~\ref{lm:ck2k} with $A = \{i: \ai \leq
  F_1\delta/2c^2k^2\}$. By the lemma, for every $j \in [h]$ we have $N_j
  \geq F_1/ck \Rightarrow \exists i \in f^{-1}(j)\text{ s.t.~} \ai \geq
  F_1\delta/(2c^2k^2)$, except with probability $\delta$. Therefore,
  $|\{j: N_j \geq F_1/ck\}| \leq \HH(c^2k^2/\delta)$. We have thus shown
  that
  $$    (1 - \beta)\HH(k) \leq 
    ( T_1(F_1/k|f) - T_1(
    F_1/ck|f) )\left(\frac{F_1}{k} - \frac{F_1}{ck}\right)^{-1}
    \leq    \HH(c^2k^2/\delta).$$
  With probability $1 - 4e^{-2\alpha}$, Lemma~\ref{lm:csum-full}
  gives us the following guarantees:
\begin{align*}
    \prob[|\tilde{T}_1(F_1/k|f) - T_1(F_1/k|f)| >
    \frac{4\alpha\sqrt{h}t}{k\epsilon}] &\leq 2\exp(-\alpha)\\
    \prob[|\tilde{T}_1(F_1/ck|f) - T_1(F_1/ck)| >
    \frac{4\alpha \sqrt{h}t}{ck\epsilon}] &\leq 2\exp(-\alpha)
\end{align*}
  With probability $1 - 4e^{2\alpha}$,
$$    
|(x_1 - x_2) - (T_1(F_1/k|f) - T_1(F_1/ck))| \leq\\
\frac{4\alpha \sqrt{h}t}{\epsilon k}\left(1 + \frac{1}{c}\right).
$$ 
A
straightforward computation and a union bound will complete the proof.
\end{proof}

In previous work Dwork et al.~\cite{dwork-pan} present an algorithm
that outputs an estimate for $\HH(k)$ in $[HH(k(1+\rho)) - \alpha*m,
HH(k/(1+rho) + \alpha*m]$ with probability $1-\delta$. Extending their
algorithm to full space can improve the additive error to
$O(\sqrt{m})$ with constant probability. For $k = O(1)$, which is the
usual range for this parameter, our algorithm outperforms Dwork et
al.'s.

%% file: extensions.tex
\section{Extensions}
\label{ext}

In the following we extend ideas used in the Heavy Hitters upper bound
to other problems. We consider the cropped second moment and inner
product problems which have not been addressed in the context of
pan-privacy before. The performence of our inner product algorithm
matches the lower bound presented in Section~\ref{LB}. We show how to
relax some of the assumptions made in Section~\ref{HH}.

\subsection{Inner Products and $T_2$}

A simple extension of the cropped sum estimator from Section~\ref{HH}
allows us to estimate the cropped dot product of two sets of 
updates, as well as the cropped second moment of an input.
%The second moment for an input $S \in [m]^n$ is defined as
%\sum_{i \in [m]}{n_i^2}$. 

Let, as before, $\mathcal{D}_0$ be the uniform distribution over $\{0,
1\}$ and $\mathcal{D}_1$ be the distribution that assigns probability
$1/2 + \epsilon/4$ to 1 and the remaining probability to 0.We compute
an estimate $\widetilde{(\va\cdot\va')}(\tau)$ of $(\va\cdot\va')(\tau)$ as
follows:
\begin{itemize}
\item For each $i \in [m]$, initialize a independentently initialized
  counters $c_i, c'_i \in_R \{0, \ldots \sqrt{\tau}-1\}$, and
  bits $b_i, b'_i \sim \mathcal{D}_0$
\item When item $i$ arrives as an update in $S$, increment the
  counter $c_{i} \pmod{\sqrt{\tau}}$. If $c_{i} = 0$ pick
  $b_{i}$ from $\mathcal{D}_1$. Process the updates in $S'$
  analogously.
\item At query time,
  \begin{itemize}
  \item compute $o := |\{i: b_i = b'_i = 1\}|$;
  \item output:
    $$
    \widetilde{(\va \cdot\va')}(\tau) = (o - \tilde{T_1}(\sqrt{\tau})/2 -
    \tilde{T_1'}(\sqrt{\tau})/2 -  m/4)\frac{16\tau}{\epsilon^2}
    $$
  \end{itemize}
\end{itemize}

\begin{lemma}\label{lm:cdot}
  The estimator $\widetilde{(\va\cdot\va')}(\tau)$ is $2\epsilon$-differentially
  private. Moreover, with probability $1 - 6e^{-2 \alpha}$,
  $$
  |(\va\cdot\va')(\tau) - \widetilde{(\va\cdot\va')}(\tau)| \leq \frac{16\alpha
    \tau\sqrt{m}}{\epsilon^2} \left(1 +
    \frac{\epsilon}{4\sqrt{\tau}}\right). 
  $$
\end{lemma}
\begin{proof}
  The proof of privacy follows from the analysis of the cropped means
  estimator~\cite{dwork-pan}. The utility analysis is
  also a simple extensions as follows.

  By the analysis of ~\cite{dwork-pan}, for every $i \in [m]$,
  $\prob[b_i = 1] = 1/2 + \epsilon \min(\ai,
  \sqrt{\tau})/4\sqrt{\tau}$, and similarly $\prob[b'_i = 1] = 1/2 +
  \epsilon \min(\ai', \sqrt{\tau})/4\sqrt{\tau}$. Since for every $i$,
  $b_i$ and $b'_i$ are independent, we have
  \begin{align*}
    \prob[b_i &= b'_i = 1] = \left(\frac{1}{2} +
      \frac{\epsilon \min(\ai,
        \sqrt{\tau})}{4\sqrt{\tau}}\right)\left(\frac{1}{2} + \frac{\epsilon
        \min(\ai', \sqrt{\tau})}{4\sqrt{\tau}}\right) \\
    &= \frac{1}{4} + \frac{\epsilon
      \min(\ai, \sqrt{\tau})}{8\sqrt{t}} +  \frac{\epsilon
      \min(\ai', \sqrt{\tau})}{8\sqrt{\tau}} + \frac{\epsilon^2
      \min(\ai \ai', \tau)}{16\tau}.
  \end{align*}

  Therefore, $\expect[\widetilde{(\va\cdot\va')}(\tau)] =
  (\va\cdot\va')(\tau)$, and the theorem follows by a Hoeffding bound
  and the guarantees for $\tilde{T_1}$.
\end{proof}

Notice that Lemma~\ref{lm:cdot} is valid regardless of whether $S_t$ and
$S'_t$ are interleved in an arbitrary manner. Also, we can take $S_t = S'_t$, and the algorithm gives an estimate for $T_2$,
i.e. $\tilde{T_2}(\tau) = \widetilde{(\va\cdot\va)}(\tau)$.

\subsection{Random Oracle and $F_1^{(t_0)}$}

Two assumptions that we make in Section~\ref{HH} are that we have
oracle access to a random function $f$ and that the value
$F_1^{(t_0)}$ at the time step $t_0$ when the algorithm is queried is
known before the sequence of updates is processed. Here we show how
these assumptions can be relaxed.

Notice that, assuming a bound $a_i \leq U$, our heavy hitters
algorithm uses constant space. Therefore Nisan's pseudorandom
generator~\cite{nisan1992pseudorandom} can be used to remove the first
assumption. To address the second assumption, we can assume an upper
bound $U_0$ on $F_1 = F_1^{(t_0)}$. Then we can run $\log U_0 + 1$
instances of our heavy hitters algorithm in parallel with $F'_1$ (the
projected value of $F_1^{(t_0)}$ set to $1, 2, 4, \ldots, U_0$,
respectively. At query time we use the output of the algorithm
instance with $F'_1$ set to $2^{\lceil \log F_1\rceil}$. This
procedure gives us a $2(\log U_0 + 1)\epsilon$-pan private algorithm
that outputs an estimate $\tilde{\HH}(k) \in [(1-\beta)\HH(k) -
O(\sqrt{k}), \HH(O(k^2)) + O(\sqrt{k})]$.

%% file: sect5.tex
\section{Concluding Remarks}
Inspired by~\cite{dwork-pan}, we study pan-private algorithms that guarantee differential privacy
of data analyses even when the internal memory of the algorithm may be compromised by an unannounced 
intrusion of an attacker. ~\cite{dwork-pan} used techniques from random response~\cite{rr} on top
of sampling to get pan-private streaming algorithms for some of the basic statistical estimates on 
the input. 

We addressed fundamental questions about the memory, its size and its role in pan-privacy. 
We showed that distinct count can not be estimated accurately to additive error even given unbounded space;
this is based on approach for showing lower bounds via noisy decoding. 
We also showed a streaming algorithm that is pan-private and matches this accuracy. We also show
worst case $O(k)$ approximate streaming pan-private algorithm for estimating heavy-hitter counts. 
Both of these upper bounds come from using sketches. Also, it is interesting that while we do not require
pan-private algorithms to use small memory, the best known algorithms so far are streaming, that is,
they use sublinear memory. 

We find the notion of pan-privacy to be intriguing, and believe more needs to be understood in this
intersection of differential privacy and streaming. For example, in streaming, many problems can be 
solved in presence negative and positive $d_j$'s. While our distince count estimation in this
paper works in this case and is pan-private, we  leave it open to address the difficulty of obtaining 
pan-private algorithms for other problems in such cases. Also, the basic model of pan-privacy here can 
be extended to the case when there are multiple intrusions or even continual intrusions~\cite{cynthia}. 
Under those models, what statistical estimates can be computed accurately and privately? 

We conclude this paper with the observation that our insights so far give pan-private 
approximations for related problems such as $T_2$ (cropped $F_2$) and inner products. We leave it
open to extend these results to other problems such as entropy estimation.